\documentclass[12pt]{amsart}
\usepackage{}

\usepackage{amsmath}
\usepackage{amsfonts}
\usepackage{amssymb}

\usepackage[all,cmtip]{xy}           
\usepackage{bm}
\usepackage{bbm}

\usepackage{bbding}
\usepackage{txfonts}
\usepackage{amscd}
\usepackage[all]{xy}
\usepackage[shortlabels]{enumitem}
\usepackage{ifpdf}
\ifpdf
  \usepackage[colorlinks,final,backref=page,hyperindex]{hyperref}
\else
  \usepackage[colorlinks,final,backref=page,hyperindex,hypertex]{hyperref}
\fi
\usepackage{tikz}
\usetikzlibrary{positioning}
\usepackage[active]{srcltx}
\usepackage{caption}

\makeatletter

\newtheorem{thm}{Theorem}[section]
\newtheorem{lem}[thm]{Lemma}
\newtheorem{cor}[thm]{Corollary}
\newtheorem{pro}[thm]{Proposition}
\theoremstyle{definition}
\newtheorem{ex}[thm]{Example}
\newtheorem{rmk}[thm]{Remark}
\newtheorem{defi}[thm]{Definition}

\newcommand {\emptycomment}[1]{}

\newcommand{\lon }{\,\rightarrow\,}
\newcommand{\be }{\begin{equation}}
\newcommand{\ee }{\end{equation}}

\newcommand{\g}{\mathfrak g}

\newcommand{\huaB}{\mathcal{B}}

\newcommand{\huaO}{{\mathcal{O}}}

\newcommand{\half}{\frac{1}{2}}

\newcommand{\Id}{{\rm{Id}}}

\newcommand{\br}[1]{   [ \cdot,    \cdot  ]   }

\newcommand{\Gr}{\mathrm{Gr}}

\newcommand{\Ad}{\mathrm{Ad}}
\newcommand{\Aut}{\mathrm{Aut}}

\newcommand{\Ker}{\mathrm{Ker}}

\newcommand{\Img}{\mathrm{Im}}

\begin{document}

\title{Rota-Baxter operators on braces, post-braces and the Yang-Baxter equation}

\author{Li Guo}
\address{Department of Mathematics and Computer Science,
         Rutgers University,
         Newark, NJ 07102}
\email{liguo@rutgers.edu}

\author{Yan Jiang}
\address{Department of Mathematics, City University of Hong Kong, Hong Kong SAR, China}
\email{yanjiang22@mails.jlu.edu.cn; yjian24@cityu.edu.hk}

\author{Yunhe Sheng}
\address{Department of Mathematics, Jilin University, Changchun 130012, Jilin, China}
\email{shengyh@jlu.edu.cn}

\author{You Wang}
\address{Department of Mathematics, Jilin University, Changchun 130012, Jilin, China}
\email{wangyou20@mails.jlu.edu.cn}

\begin{abstract}
Combining the notions of braces and relative Rota-Baxter operators on groups in connection with the Yang-Baxter equation and a factorization theorem of Lie groups from integrable systems, relative Rota-Baxter operators on braces and post-braces are introduced. A relative Rota-Baxter operator on a brace naturally induces a post-brace, and conversely, every post-brace determines a relative Rota-Baxter operator on its sub-adjacent brace. Furthermore, a post-brace yields two Drinfel'd-isomorphic solutions to the Yang-Baxter equation. As a special case, {\it enhanced} relative Rota-Baxter operators give rise to matched pairs of braces. Focusing on enhanced Rota-Baxter operators on two-sided braces, a corresponding factorization theorem is established. Examples are provided from the two-sided brace associated with the three-dimensional Heisenberg Lie algebra.
\end{abstract}

\renewcommand{\thefootnote}{}



\maketitle

\vspace{-1cm}

\tableofcontents

\vspace{-1cm}

\allowdisplaybreaks

\section{Introduction}

In this paper, relative Rota-Baxter operators on braces are introduced to connect to post-braces and matched pairs of braces, to construct solutions of the Yang-Baxter equation, and to give a factorization theorem for two-sided braces.

\subsection{Braces and set-theoretical solutions of the Yang-Baxter equation}

The Yang-Baxter equation first appeared in theoretical physics and statistical mechanics in the works of C.\,N. Yang \cite{Yang}. Independently, G. Baxter employed the Yang-Baxter equation in his solution to the eight-vertex model \cite{Baxter}. Since then, the Yang-Baxter equation has been extensively studied, leading to broad achievements in various areas of mathematics and mathematical physics.
In \cite{Drinfeld}, Drinfel'd proposed to study set-theoretical solutions of the Yang-Baxter equation on a set $X$, which is a bijective map
$ R:X \times X \to X \times X$ 
 satisfying the following equation:
$$ (R\times \Id)(\Id \times R)(R\times \Id)=(\Id \times R)(R\times \Id)(\Id \times R).  $$
In the seminal work \cite{ESS},
Etingof-Schedler-Soloviev   introduced the structure group $G(X,R)$  to classify  set-theoretical solutions.
Later on,  Lu-Yan-Zhu \cite{LYZ} showed that braided groups   give rise to set-theoretical solutions of the Yang-Baxter equation. In \cite{DRS}, Doikou-Rybolowicz-Stefanelli introduced the notion of Drinfel'd homomorphisms between set-theoretical solutions of the Yang-Baxter equation and showed that
Drinfel'd isomorphic classes of left non-degenerate solutions correspond to isomorphic classes of left
shelves endowed with special endomorphisms.

As a generalization of radical rings, braces were introduced by Rump \cite{Rump3,Rump} to produce non-degenerate involutive solutions of the Yang-Baxter equation. It was proved that braces are equivalent to bijective $1$-cocycles in \cite{CJd}. See \cite{CMS,CGO,Gateva-Ivanova2,MRS,Smoktunowicz1} for recent achievements in the study of braces and their related set-theoretical solutions to the Yang-Baxter equation.

Afterwards, skew braces were introduced by Guarnieri and Vendramin \cite{GV} as a nonabelian generalization of braces. Then the factorization problem for skew-braces was systematically studied in \cite{JKVV}. It is remarkable that skew braces are connected with other algebraic structures in numerous areas, such as matched pairs of groups, groups with exact factorizations, rings and near-rings, regular subgroups, bijective $1$-cocycles and Hopf-Galois extensions. These connections led to several new types of solutions of the Yang-Baxter equation\,\cite{SV}. Later on, Brzezi\'nski \cite{Tomasz} introduced the notion of skew trusses, as a common generalization of skew left rings and skew braces. More recently, Trappeniers \cite{Trappeniers} studied two-sided skew braces and showed that every two-sided skew brace is an extension of a weakly trivial skew brace by a two-sided brace.

\vspace{-.1cm}

\subsection{Rota-Baxter operators on groups and post-groups}

The notion of Rota-Baxter operators on associative algebras was introduced by G. Baxter
\cite{G-Baxter} in his probability study to understand Spitzers identity in fluctuation theory, and found applications in Connes-Kreimer's algebraic approach to the renormalization in quantum field theory \cite{CK}. See \cite{Guo} for more details. The notion of relative Rota-Baxter operators (also called $\huaO$-operator) on Lie algebras was introduced in \cite{Kupershmidt}, which are closely related to the classical Yang-Baxter equation and naturally induce pre-Lie or post-Lie algebras \cite{BGN,Burde}. This connection generalizes to the splitting of operads \cite{BBGN}.

The notion of Rota-Baxter operators on Lie groups was introduced in \cite{GLS}, whose differentiation give rise to Rota-Baxter operators on the corresponding Lie algebras.
The authors obtained a factorization theorem of Rota-Baxter Lie groups, realizing the Global Factorization Theorem of Semenov-Tian-Shansky on Lie groups\,\cite{STS} without going through integration.
Consequently, Goncharov defined Rota-Baxter operators on cocommutative Hopf
algebras in \cite{Go} such that many classical results still hold at the Hopf algebra level. In \cite{BaG,BaG2}, Bardakov and Gubarev discovered the relation between skew braces and Rota-Baxter groups. In \cite{BGST}, the notion of post-groups was introduced, as the derived structure of Rota-Baxter operators on groups. Post-groups are equivalent to skew braces, and lead to non-involutive solutions of the Yang-Baxter equation. Recently, the notion of Rota-Baxter type operators on trusses was introduced in \cite{CEMM} from the perspective of Rota-Baxter operators on rings.

The notion of relative Rota-Baxter operators on groups was introduced in \cite{JSZ}, in order to establish a local Lie theory   from the perspective of the local integration and differentiation.  It was shown that relative Rota-Baxter operators on groups give rise to matched pairs of groups, a notion critical to the Yang-Baxter equation\,\cite{ESS,LYZ}.
See \cite{Ta} for further details.

\vspace{-.1cm}

\subsection{Main results and outline of the paper}
In this paper, we introduce the notion of {\bf relative Rota-Baxter operators on braces}, generalizing the Rota-Baxter theory from groups to the significantly richer setting of braces that comprise of two compatible group laws. The added complexity of braces necessitates new ideas and constructions to obtain a coherent theory with expected applications to the Yang-Baxter equation.

We develop the new framework in three steps.
First, we introduce the notions of relative Rota-Baxter operators on braces and of post-braces, generalizing the classical correspondence between relative Rota-Baxter operators on groups and post-groups. An interesting phenomenon is that a post-brace  produces a pair of Drinfel'd-isomorphic solutions of the Yang-Baxter equation.
Second, we impose an {\em enhanced} condition on relative Rota-Baxter operators, and show that it provides the needed restriction to obtain matched pairs of braces that extend matched pairs of groups.
Finally, we apply enhanced Rota-Baxter operators to two-sided braces and establish a new factorization theorem for braces in the sense of \cite{JKVV}, generalizing the classical factorization theorems for Lie groups from integrable systems\,\cite{GLS, STS}.

The paper is organized as follows. In Section \ref{sec:rRBO-brace},  we introduce the notions of relative Rota-Baxter operators on braces and post-braces,  and study their relations.
Moreover, we prove that a post-brace gives rise to two set-theoretical solutions of the Yang-Baxter equation that are Drinfel'd isomorphic (Theorem\,\ref{post-brace-YBE}).
In Section \ref{sec:WRBO-MP-brace}, we define enhanced relative Rota-Baxter operators on braces and matched pairs of braces, and show that an enhanced relative Rota-Baxter operator on a brace induces a matched pair of braces (Theorem\,\ref{rRBO-mp}).
In Section \ref{sec:RBO-two-brace}, we utilize the notion of (enhanced) Rota-Baxter operators on two-sided braces to obtain a factorization theorem of two-sided braces (Theorem\,\ref{fac-thm-eRB-brace}).
As an illustration of the general results, in Section \ref{sec:enumeration}, we provide various examples based on the brace induced by the $3$-dimensional Heisenberg Lie algebra.

\section{Relative Rota-Baxter operators on braces and post-braces}\label{sec:rRBO-brace}

In this section, we first recall the semi-trivial action of braces. Then we define the notion of relative Rota-Baxter operators on braces by using the semi-trivial action. As the underlying structures of relative Rota-Baxter operators on braces, the notion of post-braces are introduced, which can give rise to   set-theoretical solutions of the Yang-Baxter equation.

\subsection{Relative Rota-Baxter operators on braces}

Let us first recall the notion of semi-trivial actions of braces. Then we introduce the notion of relative Rota-Baxter operators on braces, which naturally split the brace structures into post-brace structures.

\begin{defi} \begin{enumerate}
\item[{\rm(i)}]
A {\bf skew brace} $(G,\cdot,\circ)$ consists of a group $(G,\cdot)$ and a group $(G,\circ)$ such that
\begin{equation}\label{eq-brace}
 a\circ  (b \cdot c)=(a \circ b)\cdot a^{-1} \cdot (a \circ c), \quad \forall a,b,c\in G.
\end{equation}
Here $a^{-1}$ is the inverse of $a$ in the group $(G,\cdot)$.
\item[{\rm(ii)}] A {\bf brace} is a skew brace in which the group $(G,\cdot)$ is abelian. In addition, a brace $(G,\cdot,\circ)$ is called {\bf two-sided} if the following condition also holds
\begin{equation}\label{eq-two-sided-brace}
(a \cdot b)\circ c=(a \circ c)\cdot c^{-1} \cdot (b \circ c), \quad \forall a,b,c\in G.
\end{equation}
\item[{\rm(iii)}]
Let $(G,\cdot_G,\circ_G)$ and $(H,\cdot_H,\circ_H)$ be braces. A {\bf homomorphism} of braces from $(G,\cdot_G,\circ_G)$ to $(H,\cdot_H,\circ_H)$ is a map $\Psi:G \to H$ such that
\begin{equation}\label{homo-brace}
\Psi(a\cdot_G b)=\Psi(a) \cdot_H \Psi(b),\quad  \Psi(a\circ_G b)=\Psi(a) \circ_H \Psi(b),\quad
\forall a,b\in G.
\end{equation}
We denote by $\Aut(G)$ the set of automorphisms of the brace $(G,\cdot_G,\circ_G)$.
\end{enumerate}
\end{defi}

\begin{pro}{\rm(\cite{Vendramin})}
If $(G,\cdot,\circ)$ is a brace, then we have
\begin{equation}\label{equality-brace}
a\circ b^{-1}=a \cdot (a\circ b)^{-1} \cdot a, \quad \forall a,b\in G.
\end{equation}
Moreover, if $(G,\cdot,\circ)$ is a two-sided brace, then we also have
\begin{equation}\label{equality-two-sided-brace}
a^{-1} \circ b=b \cdot (a\circ b)^{-1} \cdot b, \quad \forall a,b\in G.
\end{equation}
\end{pro}

\begin{defi}\label{semi-trivial-act}{\rm(\cite{SV,Cindy})}
Let $(G,\cdot_G,\circ_G)$ and $(H,\cdot_H,\circ_H)$ be braces. Then a map  $\Phi:G \to \Aut(H)$ is called a {\bf semi-trivial action} of the brace $(G,\cdot_G,\circ_G)$ on the brace $(H,\cdot_H,\circ_H)$, if $\Phi(a \circ_G b)=\Phi(a)\Phi(b)$ for all $a,b\in G$.
\end{defi}

\begin{rmk}{\rm
In this case, $\Phi:G \to \Aut(H)$ is an action of the group $(G,\circ_G)$ on the group $(H,\circ_H)$.
}\end{rmk}

\begin{pro}
Let $\Phi:G \to \Aut(H)$ be a semi-trivial action of the brace $(G,\cdot_G,\circ_G)$ on the brace $(H,\cdot_H,\circ_H)$. Then $(H \times G, \cdot_{\rtimes}, \circ_{\rtimes})$ is a brace, where $\cdot_{\rtimes}$ and $\circ_{\rtimes}$ are defined by
\begin{eqnarray}
\label{semi-direct-1} (h,a)\cdot_{\rtimes} (k,b)&=&(h \cdot_H k,a \cdot_G b);\\
\label{semi-direct-2} (h,a)\circ_{\rtimes} (k,b)&=&(h \circ_H \Phi(a)k,a \circ_G b), \quad \forall a,b\in G, h,k\in H.
\end{eqnarray}
$(H \times G, \cdot_{\rtimes}, \circ_{\rtimes})$ is called the {\bf semi-direct product}
of braces $G$ and $H$, and denoted by $H \rtimes_{\Phi} G$.
\end{pro}

\begin{proof}
$(H \times G, \cdot_{\rtimes})$ is actually the direct product of the abelian groups $(G,\cdot_G)$ and $(H,\cdot_H)$. Since $\Phi(a\circ_G b)=\Phi(a)\Phi(b)$ and $\Phi(a)(h\circ_H k)=(\Phi(a)h)\circ_H (\Phi(a)k)$, $(H \times G, \circ_{\rtimes})$ is a group.
Moreover, we have
\begin{eqnarray*}
(h,a) \circ_{\rtimes} ((k,b)  \cdot_{\rtimes} (t,c))
&=& (h,a) \circ_{\rtimes} (k \cdot_H t, b \cdot_G c)\\
&=& \big(h \circ_{H} \Phi(a)(k \cdot_H t), a\circ_G (b \cdot_G c) \big)\\
&=& \big(h \circ_{H} (\Phi(a)k) \cdot_H (\Phi(a)t)), a\circ_G (b \cdot_G c) \big)\\
&=& \big((h \circ_{H} \Phi(a)k) \cdot_H h^{-1} \cdot_H (h \circ_{H} \Phi(a)t), (a\circ_G b) \cdot_G a^{-1} \cdot_G (a\circ_G c) \big)\\
&=& (h \circ_{H} \Phi(a)k, a \circ_G  b) \cdot_{\rtimes} (h^{-1},a^{-1}) \cdot_{\rtimes}
(h \circ_{H} \Phi(a)t, a \circ_G  c) \\
&=& \big((h,a) \circ_{\rtimes} (k,b)\big) \cdot_{\rtimes} (h,a)^{\cdot_{\rtimes} -1} \cdot_{\rtimes}
\big((h,a) \circ_{\rtimes} (t,c)\big),
\end{eqnarray*}
which implies that $(H \times G, \cdot_{\rtimes}, \circ_{\rtimes})$ is a brace.
\end{proof}

\begin{defi}\label{rRBO}
Let $\Phi:G \to \Aut(H)$ be a semi-trivial action of a brace $(G,\cdot_G,\circ_G)$ on a brace $(H,\cdot_H,\circ_H)$. A map $\huaB:H \to G$ is called a {\bf relative Rota-Baxter operator} on the brace $(G,\cdot_G,\circ_G)$ with respect to the semi-trivial action $\Phi:G \to \Aut(H)$, if
\begin{eqnarray}
\label{rRBO-1}\huaB(h) \cdot_G \huaB(k)&=&\huaB(h \cdot_H k);\\
\label{rRBO-2} \huaB(h) \circ_G \huaB(k) &=& \huaB( h \circ_H  \Phi(\huaB (h))(k)),\quad \forall h,k\in H.
\end{eqnarray}
\end{defi}

\begin{rmk}
Equation \eqref{rRBO-2} implies that $\huaB:H \to G$ is a relative Rota-Baxter operator on the group $(G,\circ_G)$ with respect to the action $\Phi$ of $G$ on $H$.
\end{rmk}

\begin{pro}
Let $\Phi:G \to \Aut(H)$ be a semi-trivial action of a brace $(G,\cdot_G,\circ_G)$ on a brace $(H,\cdot_H,\circ_H)$. Then $\huaB:H \to G$ is a relative Rota-Baxter operator on a brace $(G,\cdot_G,\circ_G)$ with respect to $\Phi$ if and only if the graph $\Gr(\huaB):=\{ (h,\huaB (h))| h \in H\}$ is a sub-brace of the semi-direct product $H \rtimes_{\Phi} G$.
\label{p:graph}
\end{pro}

\begin{proof}
By a direct calculation, we have
\begin{eqnarray*}
(h,\huaB (h)) \cdot_{\rtimes }(k,\huaB (k))&=& \Big(h \cdot_H k,\huaB (h) \cdot_G \huaB (k) \Big);\\
(h,\huaB (h)) \circ_{\rtimes }(k,\huaB (k))&=& \Big(h \circ_H \Phi(\huaB (h))(k),\huaB (h) \circ_G \huaB (k)\Big),
\end{eqnarray*}
for all $h,k \in H$. Then by \eqref{rRBO-1} and \eqref{rRBO-2}, it is obvious that $\huaB:H \to G$ is a relative Rota-Baxter operator on a brace $(G,\cdot_G,\circ_G)$ with respect to $\Phi$ if and only if the graph $\Gr(\huaB)$ is a sub-brace of the semi-direct product $H \rtimes_{\Phi} G$.
\end{proof}

\begin{pro}\label{descendent-brace}
Let $\huaB:H \to G$ be a relative Rota-Baxter operator on a brace $(G,\cdot_G,\circ_G)$ with respect to a semi-trivial action $\Phi:G \to \Aut(H)$. Then
$(H,\cdot_{\huaB},\circ_{\huaB})$ is a brace, called the {\bf descendent brace}, where $\cdot_{\huaB}$ and $\circ_{\huaB}$ are defined by
\begin{equation}\label{descendent}
h \cdot_{\huaB} k=h \cdot_H k, \quad h \circ_{\huaB} k=  h \circ_H  \Phi(\huaB (h))(k), \quad \forall h,k\in H.
\end{equation}
Moreover,  $\huaB:H \to G $ is a homomorphism of braces from $(H,\cdot_{\huaB},\circ_{\huaB})$ to $(G,\cdot_G,\circ_G)$.
\end{pro}

\begin{proof}
Since $\huaB:H \to G$ is a relative Rota-Baxter operator on the group $(G,\circ_G)$ with respect to the action $\Phi:G \to \Aut(H)$, by \cite[Proposition 3.5]{JSZ}, we deduce that $(H,\circ_{\huaB})$ is a group with the unit $e_H$ and $\Phi(\overline{\huaB (h)})(\bar{h})$ is the inverse of $h$ in $(H,\circ_{\huaB})$, where $\overline{\huaB (h)}$ and $\bar{h}$ are inverse elements of $\huaB (h)$ in the group $(G,\circ_G)$ and $h$ in the group $(H,\circ_H)$ respectively. Moreover, we have
\begin{eqnarray*}
h \circ_{\huaB} (k \cdot_{\huaB} t)
&=& h \circ_H \Phi(\huaB (h))(k \cdot_{\huaB} t)\\
&=& h \circ_H \big((\Phi(\huaB (h))k) \cdot_H (\Phi(\huaB (h))t)\big)\\
&=& \big(h \circ_H (\Phi(\huaB (h))k)\big) \cdot_Hh^{-1} \cdot_H  \big(h \circ_H (\Phi(\huaB (h))t)\big)\\
&=& (h \circ_{\huaB} k) \cdot_{\huaB} h^{-1} \cdot_{\huaB} (h \circ_{\huaB} t),
\end{eqnarray*}
which implies that $(H,\cdot_{\huaB},\circ_{\huaB})$ is a brace. It is obvious that $\huaB:H \to G $ is a homomorphism of braces from $(H,\cdot_{\huaB},\circ_{\huaB})$ to $(G,\cdot_G,\circ_G)$ in this case.
\end{proof}

\subsection{Post-braces}

In this subsection, we introduce the notion of post-braces and show that a relative Rota-Baxter operator on a brace splits the brace to a post-brace. First we recall the notion of post-groups and some basic results.

\begin{defi}{\rm(\cite{BGST})}
A {\bf post-group} is a group $(G,\circ)$ equipped with a multiplication $\rhd:G\times G\to G$ such that
\begin{enumerate}
\item[{\rm(i)}] for each $a\in G$, the left multiplication $$L^\rhd_a:G\to G, \quad L^\rhd_a b= a\rhd b, \quad \forall b\in G,$$
is an automorphism of the group $(G,\circ)$, that means,
\begin{equation}
  \label{eq-post-group-1} a\rhd (b\circ c)=(a\rhd b)\circ(a\rhd c),\quad \forall a, b, c\in G;
\end{equation}
\item[{\rm(ii)}] the following ``weighted" associativity holds,
\begin{equation}
\label{eq-post-group-2} \big(a\circ(a\rhd b)\big)\rhd c=a\rhd (b\rhd c), \quad \forall a, b, c\in G.
\end{equation}
\end{enumerate}
 \end{defi}

\begin{defi}{\rm(\cite{BGST})}
A {\bf homomorphism} of post-groups from $(G,\circ_G,\rhd_G)$ to $(H,\circ_H,\rhd_H)$ is a map $\Psi:G\to H$ that preserves the operations $\circ$ and $\rhd$:
\begin{equation*}
\Psi(a\circ_G b)=\Psi(a)\circ_H \Psi(b), \quad
\Psi(a\rhd_G b)=\Psi(a)\rhd_H \Psi(b),\quad \forall a, b\in G.
\end{equation*}
\end{defi}

\begin{rmk}{\rm
In \cite[Theorem 3.25]{BGST}, the authors proved that the category of post-groups (resp. pre-groups) is isomorphic to the category of skew braces (resp. braces).
}\end{rmk}

\begin{lem}{\rm(\cite{BGST})}
Let $e$ be the unit of the post-group $(G,\circ,\rhd)$. Then for all $a\in G$, we have
\begin{eqnarray}
\label{unit-post-brace-1} a \rhd e=e;\\
\label{unit-post-brace-2} e \rhd a=a.
\end{eqnarray}
\end{lem}

\begin{thm}{\rm(\cite{BGST})}\label{thm:post-group}
Let $(G,\circ,\rhd)$ be a post-group. Define $\ast:G\times G\to G$     by
\begin{eqnarray}
\label{eq:subadj-brace-bgst}  a\ast b:=a\circ (a\rhd b), \quad \forall a,b\in G.
\end{eqnarray}
\begin{enumerate}
    \item[{\rm(i)}] Then $(G,\ast)$ is a group with $e$ being the unit, and the inverse map $\dagger:G\to G$ given by
\begin{eqnarray}\label{dagger}
 a^\dagger := (L^\rhd_a)^{-1}(\bar{a}), \quad \forall a\in G,
\end{eqnarray}
where $\bar{a}$ is the inverse of $a$ with respect to the group $(G,\circ)$. The group $(G,\ast)$ is called the {\bf sub-adjacent group} of the post-group $(G,\circ,\rhd)$.
\item[{\rm(ii)}] The left multiplication $L^\rhd:G\to \Aut(G)$ is an action of the group $(G,\ast)$  on the group $(G,\circ)$.
\item[{\rm(iii)}] Let $\Psi: (G,\circ_G,\rhd_G) \to (H,\circ_H,\rhd_H)$ be a homomorphism of post-groups. Then $\Psi$ is a homomorphism of the sub-adjacent groups from $(G,\ast_G)$ to $(H,\ast_H)$.
\end{enumerate}
\end{thm}

\begin{defi}\label{post-brace}
  A {\bf post-brace} is a quadruple $(G,\cdot,\circ,\rhd)$ such that
\begin{itemize}
\item[{\rm(i)}] $(G,\cdot,\circ)$ is a brace,
\item[{\rm(ii)}] $(G,\circ,\rhd)$ is a post-group,
\item[{\rm(iii)}] $\rhd$ and $\cdot$ satisfy the following compatibility condition:
  \begin{equation}\label{eq-post-brace}
  a \rhd (b \cdot c)=(a \rhd b)\cdot (a \rhd c),\quad \forall a,b,c\in G.
  \end{equation}
\end{itemize}

A {\bf homomorphism of post-braces} from $(G,\cdot_G,\circ_G,\rhd_G)$ to $(H,\cdot_H,\circ_H,\rhd_H)$ is a map $\Psi:G \to H$ that preserves the operations $\cdot,\circ$ and $\rhd$:
\begin{equation}
\label{Post-homo}  \Psi(a\cdot_G b)=\Psi(a)\cdot_H \Psi(b), \quad \Psi(a\circ_G b)=\Psi(a)\circ_H \Psi(b), \quad \Psi(a\rhd_G b)=\Psi(a)\rhd_H \Psi(b), \ \forall a, b\in G.
\end{equation}

\end{defi}

In addition to the brace $(G,\cdot,\circ)$, a post-brace naturally gives rise to another brace on the underlying set.

\begin{thm}\label{subadj-brace}
Let $(G,\cdot,\circ,\rhd)$ be a post-brace. Define a binary operation $\ast:G\times G\to G$  by
\begin{eqnarray}
\label{eq:subadj-brace}  a\ast b&:=&a\circ (a\rhd b), \quad \forall a,b\in G.
\end{eqnarray}
\begin{itemize}
\item[{\rm(i)}]
Then $(G,\cdot,\ast)$ is a brace with $e$ being the unit, called the {\bf sub-adjacent brace} of the post-brace $(G,\cdot,\circ,\rhd)$ and denoted by $G_{\rhd}$.
\item[{\rm(ii)}] The left multiplication
   $$ L^\rhd_a:G\to G, \quad L^\rhd_a b= a\rhd b, \quad \forall a,b\in G,$$
  is a semi-trivial action of the brace $(G,\cdot,\ast)$ on the brace $(G,\cdot,\circ)$.
\item[{\rm(iii)}]
    Let $\Psi: G \to H$ be a  homomorphism of post-braces from $(G,\cdot_G,\circ_G,\rhd_G)$ to $(H,\cdot_H,\circ_H,\rhd_H)$. Then $\Psi$ is a homomorphism of the sub-adjacent braces from $(G,\cdot_G,\ast_G)$ to $(H,\cdot_H,\ast_H)$.
\end{itemize}
\end{thm}

\begin{proof}

 {\rm(i)}
  $(G,\ast)$ is a group according to Theorem \ref{thm:post-group} {\rm(i)}. By a direct calculation, we have
  \begin{eqnarray*}
  a \ast (b \cdot c)
  &=& a \circ (a \rhd (b \cdot c))\\
  &\overset{\eqref{eq-post-brace}}{=}& a \circ ((a \rhd b) \cdot (a \rhd c))\\
  &\overset{\eqref{eq-brace}}{=}& (a \circ (a \rhd b)) \cdot a^{-1}  \cdot (a \circ (a \rhd c))\\
  &=& (a \ast b) \cdot a^{-1}  \cdot (a \ast c),
  \end{eqnarray*}
  which implies that $(G,\cdot,\ast)$ is a skew brace.

 {\rm(ii)}
  By Theorem \ref{thm:post-group} {\rm(ii)}, if $(G,\circ,\rhd)$ is a post-group, then the left multiplication $ L^\rhd: (G,\ast) \to \Aut(G,\circ)$ is an action of the group $(G,\ast)$ on the group  $(G,\circ)$. Combined with \eqref{eq-post-brace}, we deduce that
   $L^\rhd_a$ is in $\Aut(G)$ for all $a\in G$. Thus, the left multiplication
   $L^\rhd: G \to \Aut(G)$ is a semi-trivial action of the brace $(G,\cdot,\ast)$ on the brace $(G,\cdot,\circ)$.

{\rm(iii)}
  By Theorem \ref{thm:post-group} {\rm(iii)}, if $\Psi: G \to H$ is a  homomorphism of post-groups, then $\Psi$ is a homomorphism of the groups from $(G,\ast_G)$ to $(H,\ast_H)$, which implies that $\Psi$ is a homomorphism of the sub-adjacent braces from $(G,\cdot_G,\ast_G)$ to $(H,\cdot_H,\ast_H)$.
\end{proof}

\begin{rmk}\label{another-skew-brace}
For any post-brace $(G,\cdot,\circ,\rhd)$, there exists another skew brace structure $(G,\circ,\ast)$, where $\ast$ is given by \eqref{eq:subadj-brace-bgst}.  See \cite[Proposition 3.22]{BGST}.
\end{rmk}

Now we establish the close relation between relative Rota-Baxter operators on braces and post-braces.

\begin{pro}\label{pro:rRBO-post-brace}
Let $\huaB:H \to G$ be a relative Rota-Baxter operator on a brace $(G,\cdot_G,\circ_G)$ with respect to a semi-trivial action $\Phi:(G,\cdot_G,\circ_G) \to \Aut(H,\cdot_H,\circ_H)$. Define a multiplication $\rhd_{\huaB}:H \times H \to H$ by
\begin{eqnarray}\label{rRBO-post-brace}
h \rhd_{\huaB} k= \Phi(\huaB (h))(k), \quad \forall h,k\in H.
\end{eqnarray}
Then $(H,\cdot_H,\circ_H,\rhd_{\huaB})$ is a post-brace, whose sub-adjacent brace is the descendant brace $(H,\cdot_{\huaB},\circ_{\huaB})$ given in Proposition \ref{descendent-brace}.

Conversely, let $(G,\cdot,\circ,\rhd)$ be a post-brace. Then the identity map $\Id:G \to G$ is a relative Rota-Baxter operator on the sub-adjacent brace $(G,\cdot,\ast)$ given in Theorem \ref{subadj-brace} with respect to the semi-trivial action $L^\rhd$ of the brace $(G,\cdot,\ast)$ on the brace $(G,\cdot,\circ)$.
\end{pro}

\begin{proof}
By \cite[Theorem 3.3]{BGST}, since $\huaB:H \to G$ be a relative Rota-Baxter operator on the group $(G,\circ_G)$ with respect to the action $\Phi$ of $(G,\circ_G)$ on $(H,\circ_H)$, then $(H,\circ_H,\rhd_{\huaB})$ is a post-group. By \eqref{rRBO-post-brace}, for all $h,k,t\in H$, we have
\begin{eqnarray*}
h \rhd_{\huaB} (k \cdot_H t)=\Phi(\huaB (h))(k \cdot_H t)=\Phi(\huaB (h))(k) \cdot_H \Phi(\huaB (h))(t)=(h \rhd_{\huaB} k) \cdot_H (h \rhd_{\huaB} t),
\end{eqnarray*}
which implies that $(H,\cdot_H,\circ_H,\rhd_{\huaB})$ is a post-brace. Moreover, by \eqref{descendent}, we have
$$ h \circ_{\huaB} k= h \circ_H  (\Phi(\huaB (h))k)=h \circ_H (h \rhd_{\huaB} k),\quad \forall  h,k\in H.   $$
Thus, $(H,\cdot_{\huaB},\circ_{\huaB})$ is the sub-adjacent brace of the post-brace $(H,\cdot_H,\circ_H,\rhd_{\huaB})$.

Now let $(G,\cdot,\circ,\rhd)$ be a post-brace. By \eqref{eq:subadj-brace}, we have
\begin{eqnarray*}
\Id(a) \cdot \Id(b)&=&\Id (a \cdot b);\\
\Id(a) \ast \Id(b)&=&\Id (a \circ  (L^\rhd_{\Id (a)} b)), \quad \forall a,b\in G,
\end{eqnarray*}
which implies that $\Id:G \to G$ is a relative Rota-Baxter operator on the brace $(G,\cdot,\ast)$ with respect to the semi-trivial action $L^\rhd$ on the brace $(G,\cdot,\circ)$.
\end{proof}

\subsection{Post-braces and the Yang-Baxter equation}

In this subsection, we prove that a post-brace provides two solutions of the Yang-Baxter equation that are Drinfel'd isomorphic. First we recall the notion of set-theoretical solutions of the Yang-Baxter equation.

\begin{defi}
Let $X$ be a set. A set-theoretical solution of the {\bf Yang-Baxter equation} on
$X$ is a bijective map $R:X\times X\to X\times X$ satisfying:
\begin{eqnarray}
R_{12}R_{23}R_{12}=R_{23}R_{12}R_{23},\,\,\,\,\text{where}\,\,\,\,R_{12}=R\times\Id_X,~R_{23}=\Id_X\times R.
\end{eqnarray}
A set $X$ with a set-theoretical solution of
the Yang-Baxter equation on $X$ is called a {\bf braided set} and
is denoted   by $(X,R)$.

Moreover, we denote by $R(a,b)=(\varphi_{a}(b),\psi_{b}(a))$ for all $a,b\in X$.
$R$  is called {\bf non-degenerate} if for all $a,b\in X$, the maps $\varphi_{a}$ and $\psi_{b}$ are bijective. $R$  is called {\bf involutive} if $R^2=\Id_{X\times X}$.
\end{defi}

\begin{defi}\label{equ-solution}
Let $(X,R)$ and $(Y,R')$ be two braided sets. A {\bf homomorphism} of braided sets from $(X,R)$ to $(Y,R')$ is a map $f:X \to Y$ such that $(f \times f)R=R' (f \times f)$.

If $f$ is bijective, then $f$ is called an {equivalence} of braided sets from $(X,R)$ to $(Y,R')$.
\end{defi}

Let us recall the following theorem, which implies that (skew) braces give rise to set-theoretical solutions of the Yang-Baxter equation. See \cite{ESS,GV,Soloviev,Vendramin} for more information.

\begin{thm}\label{skew-brace-YBE}{\rm(\cite{Rump})}
Let $(G,\cdot,\circ)$ be a brace. Then for $a,b\in G$, the map
\begin{eqnarray}\label{eq-skewbrace-YBE}
R_G:G \times G \to G \times G,\quad R_G(a,b)=(a^{-1}\cdot(a \circ b),\overline{a^{-1}\cdot(a \circ b)}\circ a \circ b),
\end{eqnarray}
is a non-degenerate involutive solution of the Yang-Baxter equation on the set $G$. Here $a^{-1}$ and $\bar{a}$ are inverse element of $a$ in the groups $(G,\cdot)$ and $(G,\circ)$ respectively.
\end{thm}

Next we recall the notion of derived solutions and Drinfel'd homomorphisms between solutions of the Yang-Baxter equation.

\begin{defi}{\rm(\cite{DRS,Soloviev})}
Let $R(a,b)=(\varphi_{a}(b),\psi_{b}(a))$ be a left non-degenerate solution of the Yang-Baxter equation on the set $X$. Then
\begin{eqnarray}\label{eq-derived-solution}
R^d:X \times X \to X \times X, \quad R^d(a,b)=(\varphi_a \psi_{\varphi_b^{-1} (a)}(b),a), \quad \forall a,b \in X
\end{eqnarray}
is also a solution, which is called the {\bf derived solution} of $R$.
\end{defi}

\begin{defi}{\rm(\cite{DRS})}
Let $(X,R)$ and $(Y,R')$ be two braided sets. A map $\omega:X \times X \to Y \times Y$ is called a {\bf Drinfel'd homomorphism} if
$$\omega~R=R'~\omega.$$
If $\omega$ is a bijection, $\omega$ is called a {\bf Drinfel'd isomorphism}. In this case, $(X,R)$ and $(Y,R')$ are called {\bf Drinfel'd isomorphic} via $\omega$ and denoted by $R \cong_{D} R'.$
\end{defi}

\begin{rmk}
The notion of Drinfel'd homomorphisms is weaker than the notion of homomorphisms of solutions introduced in Definition \ref{equ-solution}, in the sense that, if $(X,R)$ and $(Y,R')$ are homomorphic (resp. equivalent) via $f:X \to Y$, then they are Drinfel'd homomorphic (resp. Drinfel'd isomorphic) via $\omega=f \times f$. See \cite{DRS} for an example that the notion is strictly weaker as well as more details on Drinfel'd homomorphisms.
\end{rmk}

Next we consider the derived solution of $R_G$, which is induced by a brace $(G,\cdot,\circ)$ introduced by Theorem \ref{skew-brace-YBE}.

\begin{pro}{\rm (\cite{DRS})}
Let $(G,\cdot,\circ)$ be a brace and $R_G$ the induced non-degenerate solution of the Yang-Baxter equation on the set $G$ introduced by \eqref{eq-skewbrace-YBE}. Then the derived solution $R_G^d$ of $R_G$ takes the simple form as follows
$$ R_G^d(a,b)=(b,a ), \quad \forall a,b\in G. $$
\end{pro}

A post-brace $(G,\cdot,\circ,\rhd)$ contains two braces by Theorem \ref{subadj-brace}, which gives rise to two non-degenerate solutions of the Yang-Baxter equation. Next we show that these two solutions are actually Drinfel'd isomorphic.

\begin{thm}\label{post-brace-YBE}
Let $(G,\cdot,\circ,\rhd)$ be a post-brace. Then for $a,b\in G$,
\begin{eqnarray}\label{solution-1}
R_1(a,b)=(a^{-1}\cdot(a \circ b), \overline{a^{-1}\cdot(a \circ b)} \circ a \circ b),
\end{eqnarray}
and
\begin{eqnarray}\label{solution-2}
R_2(a,b)=(a^{-1}\cdot(a \ast b), (a^{-1}\cdot(a \ast b))^\dagger \ast a \ast b),
\end{eqnarray}
are two non-degenerate involutive solutions of the Yang-Baxter equation on the set $G$.
Here $a^\dagger$ is the inverse element of $a$ with respect to the group $(G,\ast)$ given by \eqref{dagger}.
Moreover, $R_1$ is Drinfel'd isomorphic to $R_2$ via the map $\overline{\omega}:G \times G \to G \times G$ given by
$$ \overline{\omega}(a,b)=\big(a,(L^\blacktriangleright_a)^{-1}(a^{-1}\cdot (a \circ b)) \big),\quad \forall a,b \in G,$$
where $\blacktriangleright$ is the $\lambda$-map of the brace $(G,\cdot,\ast)$ defined by $a \blacktriangleright b=a^{-1} \cdot (a \ast b)$.
\end{thm}

\begin{proof}
Let $(G,\cdot,\circ,\rhd)$ be a post-brace. By Theorem \ref{subadj-brace}, $(G,\cdot,\circ)$ and $(G,\cdot,\ast)$ are braces. Then by Theorem \ref{skew-brace-YBE}, $R_1$ and $R_2$ are non-degenerate solutions of the Yang-Baxter equation on the set $G$. By \cite[Lemma 2.12]{DRS},
$R_1$ is Drinfel'd isomorphic to its derived solution $\tau$, where $\tau$ is the flip map defined by $\tau(a,b)=(b,a)$ for all $a,b\in G$. More precisely, there exists a bijective map $\omega_1:G \times G \to G \times G$ given by $\omega_1(a,b)=(a^{-1} \cdot (a \circ b),a)$ such that $\omega_1 R_1=\tau \omega_1$. Applying the same method, there exists another bijective map
$\omega_2:G \times G \to G \times G$ given by $\omega_2(a,b)=(a \blacktriangleright b,a)$ such that $\omega_2 R_2=\tau \omega_2$.

Therefore, there exists a map $\overline{\omega}:=\omega_2^{-1} \omega_1$ such that $\overline{\omega} R_1=R_2 \overline{\omega}$, which implies that $R_1$ are Drinfel'd isomorphic to $R_2$ via the map $\overline{\omega}$.
\end{proof}

Since a relative Rota-Baxter operator on a brace naturally induces a post-brace, it naturally provides two Drinfel'd isomorphic solutions of the Yang-Baxter equation as follows.

\begin{cor}
Let $\huaB:H \to G$ be a relative Rota-Baxter operator on a brace $(G,\cdot_G,\circ_G)$ with respect to a semi-trivial action $\Phi:(G,\cdot_G,\circ_G) \to \Aut(H,\cdot_H,\circ_H)$. Then $R,R^{\huaB}:H\times H \to H \times H$ defined by
$$ R (h,k)=(h^{-1} \cdot_H (h \circ_H k), \overline{(h^{-1} \cdot_H (h \circ_H k))} \circ_{H} h \circ_{H} k ), \quad \forall h,k\in H, $$
and
$$ R^{\huaB}(h,k)=(h^{-1} \cdot_H (h \circ_{\huaB} k), (h^{-1} \cdot_H (h \circ_{\huaB} k))^\dagger \circ_{\huaB} h \circ_{\huaB} k ), \quad \forall h,k\in H, $$
are two Drinfel'd isomorphic solutions of the Yang-Baxter equation on the set $H$. Here $\overline{a}$ is the inverse element of $a$ with respect to the group $(H,\circ_{H})$, and $a^\dagger$ is the inverse element of $a$ with respect to the group $(H,\circ_{\huaB})$ given by \eqref{descendent}.
\end{cor}

\section{Enhanced relative Rota-Baxter operators and matched pairs of braces}\label{sec:WRBO-MP-brace}

In this section, we introduce the notion of enhanced relative Rota-Baxter operators on braces, as a strengthened form of relative Rota-Baxter operators. We show that an enhanced relative Rota-Baxter operator on a brace give rise to a  factorization of brace. Moreover, we also introduce the notion of matched pairs of braces and show that an enhanced relative Rota-Baxter operator on a brace gives rise to a matched pair of braces.

\begin{defi}
Let $\Phi:G \to \Aut(H)$ be a semi-trivial action of the brace $(G,\cdot_G,\circ_G)$ on the brace $(H,\cdot_H,\circ_H)$. Then a map $\huaB:H \to G$ is called an {\bf enhanced relative Rota-Baxter operator} on the brace $(G,\cdot_G,\circ_G)$ with respect to the semi-trivial action $\Phi:G \to \Aut(H)$ if
\begin{eqnarray}
\label{eh-rRBO-1}\huaB(h) \cdot_G \huaB(k)&=&\huaB(h \cdot_H k);\\
\label{eh-rRBO-2}(\huaB(h) \cdot_G a) \circ_G \huaB(k) &=& \huaB \Big( h\circ_H  \Phi(\huaB (h) \cdot_G a)k \Big) \cdot_G a,
\end{eqnarray}
for all $a\in G,h,k\in H$.
\end{defi}

\begin{rmk}
When we choose $a$ to be the unit $e_G$ of the brace $(G,\cdot_G,\circ_G)$, then \eqref{eh-rRBO-2} degenerates to \eqref{rRBO-2}, which implies that an enhanced relative Rota-Baxter operator is a special relative Rota-Baxter operator.
\end{rmk}

\begin{lem}
Let $\huaB:H\to G$ be an enhanced relative Rota-Baxter operator on a brace $(G,\cdot_G,\circ_G)$ with respect to a semi-trivial action $\Phi:G \to \Aut(H)$. Then for all $a\in G, k\in H$, we have
\begin{equation}\label{enhance-property}
a \circ_G \huaB(k)=\huaB(\Phi(a)k)\cdot_G a.
\end{equation}
\end{lem}

\begin{proof}
By \eqref{eh-rRBO-1}, we have $\huaB(e_H)=e_G$. Setting $h$ equal to $e_H$ in \eqref{eh-rRBO-2}, we have $a \circ_G \huaB(k)=\huaB(\Phi(a)k)\cdot_G a$.
\end{proof}

\begin{defi}\label{fac-skew-brace}{\rm(\cite{JKVV})}
A {\bf left ideal} of a skew brace $(G,\cdot,\circ)$ is a subgroup I of $(G,\cdot)$ such that $a \rhd I\subseteq I$ for all $a\in G$, where $\rhd$ is given by $a \rhd b= a^{-1}\cdot(a \circ b)$.
In particular, if $(I,\cdot)$ is normal in $(G,\cdot)$ and $(I,\circ)$ is normal in $(G,\circ)$, then I is called an {\bf ideal} of a skew brace $(G,\cdot,\circ)$.

Moreover, we say that {\bf $G$ admits a factorization through} $H$ and $K$, if $H,K$ are two left ideals of a skew brace $G$ and $G=H\cdot K$ as the direct product.
\end{defi}

\begin{rmk}
If I is an ideal of a skew brace $(G,\cdot,\circ)$, then we have $a\circ b=a\cdot (a\rhd b)\in I$ for all $a,b\in I$, which implies that $(I,\circ)$ is a subgroup of $(G,\circ)$ . Consequently, $(I,\cdot,\circ)$ is a sub-brace of $(G,\cdot,\circ)$.
\end{rmk}

Since $a\circ b=a \cdot (a \rhd b)$, if $G$ admits a factorization through $H$ and $K$, we have
$$ G=H\cdot K=H\circ K.$$

Now we show that enhanced relative Rota-Baxter operators on braces give rise to   factorizations of braces.
Let $(G,\cdot_G,\circ_G)$ and $(H,\cdot_H,\circ_H)$ be braces and $\huaB:H \to G$ be a map.
For all $h\in H,a\in G$, the map
\begin{eqnarray}
\xi_\huaB:H\times G\to H\rtimes_{\Phi}G, \quad  \xi_\huaB(h,a)=(h,\huaB(h)\cdot_G a),
\end{eqnarray}
is invertible. In fact, the inverse map $\xi^{-1}:H\rtimes_{\Phi}G\to H\times G$  is given by
\begin{eqnarray}
\xi_\huaB^{-1}(h,a)=(h,\huaB(h)^{-1} \cdot_G a), \quad \forall h\in H,~a\in G.
\end{eqnarray}
Transporting the brace structure on $(H\rtimes_{\Phi}G,\cdot_{\rtimes},\circ_{\rtimes})$ to $H\times G$, we obtain a brace $(H\times G,\bullet,\star)$, where the multiplications $\bullet$ and $\star$ are given by
\begin{eqnarray}
\label{bullet}(h,a)\bullet(k,b)
&=&\xi_\huaB^{-1}\Big(\xi_\huaB(h,a)\cdot_{\rtimes} \xi_\huaB(k,b) \Big)\\
\nonumber &=&\xi_\huaB^{-1}\Big((h,\huaB(h)\cdot_G a)\cdot_{\rtimes} (k,\huaB(k)\cdot_G b)\Big)\\
\nonumber &=&\xi_\huaB^{-1}\Big(h \cdot_H k,(\huaB(h)\cdot_G a)\cdot_G (\huaB(k)\cdot_G b) \Big)\\
\nonumber &=&\Big ( h \cdot_H k,\huaB(h \cdot_H k)^{-1} \cdot_G  \huaB(h)\cdot_G a \cdot_G \huaB(k)\cdot_G b  \Big),
\end{eqnarray}
and
\begin{eqnarray}
\label{star}&&(h,a)\star(k,b)\\
\nonumber &=&\xi_\huaB^{-1}\Big(\xi_\huaB(h,a)\circ_{\rtimes} \xi_\huaB(k,b) \Big)\\
\nonumber &=&\xi_\huaB^{-1}\Big((h,\huaB(h)\cdot_G a)\circ_{\rtimes} (k,\huaB(k)\cdot_G b)\Big)\\
\nonumber &=&\xi_\huaB^{-1}\Big(h\circ_H \Phi(\huaB(h)\cdot_G a)(k),(\huaB(h)\cdot_G a)\circ_G(\huaB(k)\cdot_G b) \Big)\\
\nonumber &=&\Big(h\circ_H \Phi(\huaB(h)\cdot_G a)(k), \huaB(h\circ_H \Phi(\huaB(h)\cdot_G a)k)^{-1} \cdot_G ((\huaB(h)\cdot_G a)\circ_G (\huaB(k)\cdot_G b)) \Big).
\end{eqnarray}
Moreover, the unit of the brace $(H\times G,\bullet,\star)$ is $(e_H,\huaB(e_H)^{-1})$ and the inverse elements of $(h,a)$ in the groups $(H\times G,\bullet)$ and $(H\times G,\star)$ are given by
\begin{eqnarray}
\label{inverse-big1} (h,a)^{-1} &=& \Big(h^{-1},\huaB(h^{-1})^{-1} \cdot_G a^{-1} \cdot_G  \huaB(h)^{-1} \cdot_G \huaB(e_H)  \Big);\\
\label{inverse-big2}(h,a)^{\dagger}&=&\Big(\Phi(\huaB(h)\cdot_G a)^{-1}(\bar{h}),
\huaB(\Phi(\huaB(h)\cdot_G a)^{-1}(\bar{h}))^{-1} \cdot_G \overline{(\huaB(h) \cdot_G a)} \Big).
\end{eqnarray}

\begin{thm}\label{pro:factor}
With the  above notations, $(H\times G,\bullet,\star)$  is a brace factorization into ideals $H\times\{e_G\}$ and $\{e_H\}\times G$  if and only if $\huaB:H\to G$ is an enhanced relative Rota-Baxter operator on the brace $(G,\cdot_G,\circ_G)$ with respect to the semi-trivial action $\Phi:G \to \Aut(H)$.
\end{thm}

\begin{proof}
By \eqref{bullet} and \eqref{star}, it is obvious that $\{e_H\}\times G$ is a left ideal of the brace $(H\times G,\bullet,\star)$ by a direct calculation. For all $h,k\in H$, we have
\begin{eqnarray}
\label{descendent-group2} (h,e_G)\bullet(k,e_G)&\stackrel{\eqref{bullet}}{=}& \Big ( h \cdot_H k,\huaB(h \cdot_H k)^{-1} \cdot_G  \huaB(h) \cdot_G \huaB(k) \Big);\\
\label{descendent-group}~~(h,e_G)\star(k,e_G)&\stackrel{\eqref{star}}{=}&\Big(h\circ_H\Phi(\huaB(h))
(k),\huaB(h\circ_H \Phi(\huaB(h))k)^{-1} \cdot_G (\huaB(h)\circ_G \huaB(k)) \Big).
\end{eqnarray}
Thus, $H\times\{e_G\}$ is a sub-brace of $(H\times G,\bullet,\star)$ if and only if
\begin{eqnarray*}
\huaB(h \cdot_H k)^{-1} \cdot_G  \huaB(h) \cdot_G \huaB(k) &=& e_G;\\
\huaB(h\circ_H \Phi(\huaB(h))k)^{-1} \cdot_G (\huaB(h)\circ_G \huaB(k)) &=& e_G,
\end{eqnarray*}
which implies  that $\huaB:H\to G$ is a relative Rota-Baxter operator on the brace $(G,\cdot_G,\circ_G)$. Then we have $\huaB(e_H)=e_G$ and \eqref{inverse-big1} can be simplified to the following form:
$$ (h,a)^{-1}=(h^{-1},a^{-1}), \quad \forall h\in H,~a\in G. $$
Moreover, for all $a\in G,h,k\in H$, we have
\begin{eqnarray*}
&&(h,a)^{-1} \bullet ((h,a) \star (k,e_G))\\
&\stackrel{\eqref{star}}{=}&(h^{-1},a^{-1}) \bullet \Big(h\circ_H \Phi(\huaB(h)\cdot_G a)(k), \huaB(h\circ_H \Phi(\huaB(h)\cdot_G a)k)^{-1} \cdot_G ((\huaB(h)\cdot_G a)\circ_G \huaB(k))
\Big)\\
&\stackrel{\eqref{bullet}}{=}& \Big( h^{-1} \cdot_H  (h\circ_H \Phi(\huaB(h)\cdot_G a)(k)),    a^{-1} \cdot_G \huaB(h\circ_H \Phi(\huaB(h)\cdot_G a)k)^{-1} \cdot_G ((\huaB(h)\cdot_G a)\circ_G \huaB(k)) \Big).
\end{eqnarray*}
Thus, $H\times\{e_G\}$ is a left ideal of the brace $(H\times G,\bullet,\star)$ if and only if
$$ a^{-1} \cdot_G \huaB(h\circ_H \Phi(\huaB(h)\cdot_G a)k)^{-1} \cdot_G ((\huaB(h)\cdot_G a)\circ_G \huaB(k))=e_G, $$
which implies that $\huaB:H\to G$ is an enhanced relative Rota-Baxter operator on the brace $(G,\cdot_G,\circ_G)$ with respect to the semi-trivial action $\Phi$.
\end{proof}

Next we recall the notion of a matched pair of groups. Then we introduce the notion of a matched pair of braces using matched pairs of groups.

\begin{defi}{\rm(\cite{Ta})}
A {\bf matched pair of groups} is a triple $(G,H,\sigma)$, where $(G,\cdot_G)$ and $(H,\cdot_H)$ are groups and
$$\sigma:G\times H\lon H\times G,\quad (a,h)\mapsto(a\rightharpoonup h,a\leftharpoonup h),$$
is a map satisfying the following conditions:
\begin{eqnarray}
\label{MG-1}e_G\rightharpoonup h&=&h,\\
\label{MG-2}a\rightharpoonup(b\rightharpoonup h)&=&(a \cdot_G b)\rightharpoonup h,\\
\label{MG-3}(a\cdot_G b)\leftharpoonup h&=&\big(a\leftharpoonup(b\rightharpoonup h)\big)\cdot_G (b\leftharpoonup h),\\
\label{MG-4}a\leftharpoonup e_H&=& a,\\
\label{MG-5}(a\leftharpoonup h)\leftharpoonup k&=&a\leftharpoonup(h\cdot_H k),\\
\label{MG-6}a\rightharpoonup(h\cdot_H k)&=&(a\rightharpoonup h)\cdot_H \big((a\leftharpoonup h)\rightharpoonup k\big),
\end{eqnarray}
for all   $a,b\in G$,~$h,k\in H.$
\end{defi}

\begin{pro}\label{double-group}{\rm(\cite{Ta})}
Let $(G,\cdot_G)$ and $(H,\cdot_H)$ be groups. Then $(G,H,\sigma)$ is a matched pair of groups if and only if $(H \times G,\bowtie)$ is a group with the unit $(e_H,e_G)$, where the multiplication $\bowtie$ is given by
$$ (h,a)\bowtie (k,b)=(h \cdot_H (a\rightharpoonup k),(a\leftharpoonup k)\cdot_G b), \quad \forall a,b \in G,h,k \in H. $$
$(H \times G,\bowtie)$ is called the {\bf double group} of the matched pair of groups $(G,H,\sigma)$.
\end{pro}

\begin{rmk}
Let $(G,H,\sigma)$ be a matched pair of groups. Then $\sigma$ is bijective. Moreover, the triple $(H,G,\sigma^{-1})$ is also a  matched pair of groups.
\end{rmk}

\begin{defi}
A {\bf matched pair of braces} is a quadruple $(G,H,\sigma,\theta)$, where $G=(G,\cdot_G,\circ_G)$ and $H=(H,\cdot_H,\circ_H)$ are braces and
\begin{eqnarray*}
&&\sigma:G\times H\lon H\times G,\quad (a,h)\mapsto(a\rightharpoonup h,a\leftharpoonup h);\\
&&\theta:G\times H\lon H\times G,\quad (a,h)\mapsto(a\rightharpoondown h,a\leftharpoondown h),
\end{eqnarray*}
are maps such that
\begin{itemize}
\item[{\rm(i)}] $(G,H,\sigma)$ is a matched pair of groups $(G,\cdot_G)$ and $(H,\cdot_H)$;
\item[{\rm(ii)}] $(G,H,\theta)$ is a matched pair of groups $(G,\circ_G)$ and $(H,\circ_H)$;
\item[{\rm(iii)}] $\rightharpoonup,\leftharpoonup$ and $\rightharpoondown,\leftharpoondown$ satisfy the following compatibility conditions:
\begin{eqnarray}
\label{compatible-1}
&&(a\leftharpoondown(k\cdot_{H}(b\rightharpoonup t)))\circ_{G}((b\leftharpoonup t)\cdot_G c)\\
\nonumber &=& \Big(((((a \leftharpoondown k)\circ_{G} b) \leftharpoonup a^{-1})\cdot_{G} a^{-1}) \leftharpoonup
(a\rightharpoondown t) \Big) \cdot_{G} ((a \leftharpoondown t)\circ_{G} c);\\
\label{compatible-2} &&h \circ_{H} (a \rightharpoondown (k \cdot_H (b\rightharpoonup t)) )\\
\nonumber &=& (h \circ_H (a \rightharpoondown k))\cdot_{H} \Big(((a \leftharpoondown k)\circ_{G} b)\rightharpoonup  (a^{-1} \rightharpoonup h^{-1}) \Big) \cdot_H \\
\nonumber &&  \Big( (((a \leftharpoondown k)\circ_{G} b)\leftharpoonup (a^{-1} \rightharpoonup h^{-1}))\cdot_G (a \leftharpoonup (a^{-1}\rightharpoonup h^{-1}))^{-1} \Big) \rightharpoonup
 (h \circ_{H}  (a \rightharpoondown t)),
\end{eqnarray}
for all $a,b,c\in G, h,k,t \in H$.
\end{itemize}
\end{defi}

\begin{rmk}
Let $(G,H,\sigma,\theta)$ be a matched pair of braces. Then $\sigma,\theta$ are bijective. Moreover, the quadruple $(H,G,\sigma^{-1},\theta^{-1})$ is also a  matched pair of braces.
\end{rmk}

\begin{pro}\label{double-brace}
Let $(G,\cdot_G,\circ_G)$ and $(H,\cdot_H,\circ_H)$ be braces. Then $(G,H,\sigma,\theta)$ is a matched pair of braces if and only if $(H \times G,\cdot_{\bowtie},\circ_{\bowtie})$ is a brace with unit $(e_H,e_G)$ and multiplications $\cdot_{\bowtie},\circ_{\bowtie}$ given by
\begin{eqnarray*}
(h,a)\cdot_{\bowtie} (k,b)&=&(h \cdot_H (a\rightharpoonup k),(a\leftharpoonup k)\cdot_G b), \\
(h,a)\circ_{\bowtie} (k,b)&=&(h \circ_H (a\rightharpoondown k),(a\leftharpoondown k)\circ_G b),\quad \forall a,b \in G,h,k \in H.
\end{eqnarray*}
$(H \times G,\cdot_{\bowtie},\circ_{\bowtie})$ is called the {\bf double} brace of the matched pair of braces $(G,H,\sigma,\theta)$.
\end{pro}

\begin{proof}
By Proposition \ref{double-group}, we deduce that $(G,H,\sigma)$ is a matched pair of groups if and only if $(H \times G,\cdot_{\bowtie})$ is a group with unit $(e_H,e_G)$. Similarly, $(G,H,\theta)$ is a matched pair of groups if and only if $(H \times G,\circ_{\bowtie})$ is a group with unit $(e_H,e_G)$. Moreover, by a direct calculation, for all $a,b,c\in G,h,k,t\in H$,
$$ (h,a) \circ_{\bowtie} ((k,b) \cdot_{\bowtie} (t,c))=(h,a) \circ_{\bowtie} (k,b) \cdot_{\bowtie} (h,a)^{\cdot_{{\bowtie}} -1} \cdot_{\bowtie}  (h,a) \circ_{\bowtie} (t,c) $$
if and only if the equations \eqref{compatible-1} and \eqref{compatible-2} hold. Therefore,
$(G,H,\sigma,\theta)$ is a matched pair of braces if and only if $(H \times G,\cdot_{\bowtie},\circ_{\bowtie})$ is a brace.
\end{proof}

\begin{rmk}
The notion of a matched pair of braces here is very different from the one introduced in \cite{Bachiller2}. The latter was used to construct a non-trivial example of finite simple left braces and recover a finite left brace from its Sylow subgroups.
\end{rmk}

Enhanced relative Rota-Baxter operators on braces naturally give rise to matched pairs of braces.

\begin{thm}\label{rRBO-mp}
Let $\huaB:H \to G$ be an enhanced relative Rota-Baxter operator on a brace $(G,\cdot_G,\circ_G)$ with respect to a semi-trivial action $\Phi:G \to \Aut(H)$. Define $\rightharpoonup,\rightharpoondown:G \times H \to H$ and $\leftharpoonup,\leftharpoondown:G \times H \to G$ respectively by
\begin{eqnarray*}
a \rightharpoonup h&=&h, \quad\quad\quad  a \leftharpoonup h=a,\\
a \rightharpoondown h&=& \Phi(a)h, \quad a \leftharpoondown h=a,
\end{eqnarray*}
for all $a\in G,h\in H$. Then $(G,H,\sigma,\theta)$ is a matched pair of braces, where
$\sigma(a,h)=(a \rightharpoonup h,a\leftharpoonup h)$ and $\theta(a,h)=(a \rightharpoondown h,a\leftharpoondown h).$
\end{thm}

\begin{proof}
By Theorem \ref{pro:factor}, if $\huaB:H \to G$ is an enhanced relative Rota-Baxter operator on a brace $(G,\cdot_G,\circ_G)$ with respect to $\Phi:G \to \Aut(H)$, $(H\times G,\bullet,\star)$  has a brace factorization into ideals $H\times\{e_G\}$ and $\{e_H\}\times G$. By \eqref{bullet} and \eqref{star}, for all $h\in H,a\in G$, we have
\begin{eqnarray*}
 (e_H,a) \bullet (h,e_G)&=&(h,a);\\
 (e_H,a) \star (h,e_G)&=& \big(\Phi(a)h,\huaB(\Phi(a)h)^{-1} \cdot_{G} (a \circ_G \huaB(h)) \big)\overset{\eqref{enhance-property}}{=}(\Phi(a)h,a),
\end{eqnarray*}
which implies that $a \rightharpoonup h=h,~a \leftharpoonup h=a$ and
$a \rightharpoondown h= \Phi(a)h,~a \leftharpoondown h=a$ by Proposition \ref{double-brace}.
It is obvious that $(G,H,\sigma)$ is a matched pair of groups $(G,\cdot_G)$ and $(H,\cdot_H)$ and $(G,H,\theta)$ is also a matched pair of groups $(G,\circ_G)$ and $(H,\circ_H)$. We only need to prove that $\rightharpoonup,\leftharpoonup,\rightharpoondown,\leftharpoondown$ satisfies the compatibility conditions \eqref{compatible-1} and \eqref{compatible-2}. For all $a\in G,h,k,t\in H$, we have
\begin{eqnarray*}
h \circ_H \Phi(a)(k \cdot_H t)= h \circ_H (\Phi(a)(k) \cdot_H  \Phi(a)(t))
=  (h \circ_H \Phi(a)(k)) \cdot_H h^{-1} \cdot_H (h \circ_H \Phi(a)(t)),
\end{eqnarray*}
which implies that Eq. \eqref{compatible-2} holds. Eq. \eqref{compatible-1} holds naturally by \eqref{eq-brace}. Thus, $(G,H,\sigma,\theta)$ is a matched pair of braces.
\end{proof}

\section{Enhanced Rota-Baxter operators on two-sided braces and factorization theorems}\label{sec:RBO-two-brace}

In this section, we introduce the notion of Rota-Baxter operators and enhanced Rota-Baxter operators on two-sided braces. An enhanced Rota-Baxter operator on a two-sided brace gives rise to a braided brace and possesses a factorization theorem, generalizing the factorization theorem of Lie groups given by \cite[Theorem 3.5]{GLS}.

\begin{pro}
Let $(G,\cdot,\circ)$ be a two-sided brace. Then the adjoint action $\Ad^{\circ}$ for the group
$(G,\circ)$ is a semi-trivial action of the two-sided brace $(G,\cdot,\circ)$ on itself.
\end{pro}

\begin{proof}
We only need to show that $\Ad^{\circ}_a$ is an automorphism for the group $(G,\cdot)$ for all $a\in G$. Actually, for all $a,b,c\in G$, we have
\begin{eqnarray*}
\Ad^{\circ}_{a}(b \cdot c)&=& a\circ (b \cdot c) \circ \bar{a}\\
&=& (a \circ b \cdot a^{-1} \cdot a \circ c)\circ \bar{a}\\
&\overset{\eqref{eq-two-sided-brace}}{=}& (a \circ b \circ \bar{a})\cdot \bar{a}^{-1} \cdot
(a^{-1} \circ \bar{a}^{-1})\cdot \bar{a}^{-1} \cdot (a \circ c \circ \bar{a})\\
&\overset{\eqref{equality-two-sided-brace}}{=}& (a \circ b \circ \bar{a})\cdot (a \circ \bar{a})\cdot (a \circ c \circ \bar{a})\\
&=& \Ad^{\circ}_{a}(b) \cdot \Ad^{\circ}_{a}(c),
\end{eqnarray*}
which implies that $\Ad^{\circ}_a \in \Aut(G,\cdot)$ for all $a\in G$. Thus, $\Ad^{\circ}$ is a semi-trivial action of the two-sided brace $(G,\cdot,\circ)$ on itself.
\end{proof}

Now we can define (enhanced) Rota-Baxter operators on two-sided braces.
\begin{defi}
Let $(G,\cdot,\circ)$ be a two-sided brace. Then a map $\huaB:G \to G$ is called a {\bf  Rota-Baxter operator} on the two-sided brace $(G,\cdot,\circ)$, if
\begin{eqnarray}
\label{RBO-1}\huaB(a) \cdot \huaB(b)&=&\huaB(a \cdot b);\\
\label{RBO-2} \huaB(a) \circ \huaB(b) &=& \huaB( a \circ (\Ad^{\circ}_{\huaB (a)}b)),
\end{eqnarray}
for all $a,b\in G$. If in addition, $\huaB:G \to G$ satisfies the equality
\begin{eqnarray}\label{enhanced-RBO}
\label{en-RBO} (\huaB(a) \cdot x)\circ \huaB(b) &=& \huaB( a \circ (\Ad^{\circ}_{\huaB (a) \cdot x}b)) \cdot x, \quad \forall a,b,x\in G,
\end{eqnarray}
then $\huaB:G \to G$ is called an {\bf enhanced Rota-Baxter operator} on the two-sided brace $(G,\cdot,\circ)$. An {\bf enhanced Rota-Baxter two-sided brace} is a two-sided brace $(G,\cdot,\circ)$ together with an enhanced Rota-Baxter operator $\huaB$.
\end{defi}

\begin{lem}
Let $\huaB:G \to G$ be an enhanced Rota-Baxter operator on the two-sided brace $(G,\cdot,\circ)$.
Then we have
\begin{eqnarray}\label{enhance-RBO}
b \circ \huaB(b)=\huaB(b) \cdot b, \quad \forall b\in G.
\end{eqnarray}
\end{lem}

\begin{proof}
By \eqref{RBO-1}, we have $\huaB(e_G)=e_G$. Setting $a$ equal to $e_G$ and $x$ equal to $b$ in \eqref{en-RBO}, then we have $b \circ \huaB(b)=\huaB(b) \cdot b$ for all $b\in G$.
\end{proof}

By Proposition \ref{descendent-brace}, we obtain that a Rota-Baxter operator on the two-sided brace induces a descendent brace. However, the descendent brace might no longer be two-sided.

\begin{cor}\label{two-descendent-brace}
If $\huaB:G \to G$ is a Rota-Baxter operator on the two-sided brace $(G,\cdot,\circ)$, then
$(G,\cdot,\ast)$ is a brace, called the {\bf descendent brace} and denoted by $G_{\huaB}$, where $\ast$ is defined by
\begin{equation}\label{two-sided-descendent}
a \ast b=  a \circ  \Ad^{\circ}_{\huaB (a)}(b), \quad \forall a,b\in G.
\end{equation}
Moreover,  $\huaB:(G,\cdot,\ast) \to (G,\cdot,\circ) $ is a homomorphism of braces.
\end{cor}

Next we prove a factorization theorem of enhanced Rota-Baxter two-sided braces, which generalizes the factorization theorem of Rota-Baxter groups given by \cite[Theorem 3.5]{GLS}.

Let $(G,\huaB)$ be an enhanced Rota-Baxter two-sided brace. Denote by $(G_{\huaB},\cdot,\ast)$ the descendent brace given in Corollary \ref{two-descendent-brace}.

\begin{pro}\label{huaB+}
Let $(G,\huaB)$ be an enhanced Rota-Baxter two-sided brace. Define
$$ \huaB_+:G \to G,\quad \huaB_+(a)=a \circ \huaB(a). $$
Then $\huaB_+$ is a homomorphism of braces from $G_{\huaB}$ to $G$.
\end{pro}

\begin{proof}
For all $a,b\in G$, we have
\begin{eqnarray*}
\huaB_+(a \cdot b)
&=& (a \cdot b) \circ \huaB(a \cdot b)\\
&\overset{\eqref{enhance-RBO}}{=}& \huaB(a \cdot b) \cdot (a \cdot b)\\
&=& (\huaB(a) \cdot a) \cdot(\huaB(b) \cdot b)\\
&=& \huaB_+(a) \cdot \huaB_+(b),\\
\end{eqnarray*}
and
\begin{eqnarray*}
\huaB_+(a \ast b)
&=& a \circ (\Ad^{\circ}_{\huaB(a)} b)\circ \huaB(a \circ \Ad^{\circ}_{\huaB(a)} b)\\
&=& a \circ \huaB(a) \circ b \circ \overline{\huaB(a)} \circ \huaB(a)\circ \huaB(b)\\
&=& a \circ \huaB(a) \circ b \circ \huaB(b)\\
&=& \huaB_+(a) \circ \huaB_+(b).
\end{eqnarray*}
Thus, $\huaB_+$ is a homomorphism of braces from $G_{\huaB}$ to $G$.
\end{proof}

Let $\huaB$ be an enhanced Rota-Baxter operator on $G$. Define four subsets of $G$ as follows:
$$ G_+:=\Img \huaB_+,\quad G_-:=\Img \huaB,\quad K_+:=\Ker \huaB,\quad K_-:=\Ker \huaB_+.$$
Since both $\huaB$ and $\huaB_+$ are homomorphisms of braces, it follows that $G_+$ and $G_-$ are sub-braces of G,
$K_+$ and $K_-$ are ideals of $G_{\huaB}$ and $G_{\pm} \cong G_{\huaB}/ K_{\mp}$ as braces.
Moreover, we have the following relations.

\begin{lem}\label{lem-fac}
$K_+ \subset G_+$ and $K_- \subset G_-$ are ideals of braces.
\end{lem}

\begin{proof}
Let $a\in K_-$, that is, $\huaB_+(a)=a \circ \huaB(a)=e$. Then we have $a=\overline{\huaB(a)}=\huaB(a^{\dagger})$, where $a^{\dagger}=\Ad^{\circ}_{\overline{\huaB (a)}} \bar{(a)} $ is the inverse of $a$ in the group $(G_{\huaB},\ast)$. Thus, $a\in G_-$ and hence $K_- \subset G_-$.

For any $a \in K_{-}$ and $\huaB(b) \in G_{-}$, let us check $\huaB(b) \circ a \circ \overline{\huaB(b)} \in K_{-}$. In fact, we have
\begin{eqnarray*}
b * a * b^{\dagger} &=& b\circ (\Ad^{\circ}_{\huaB(b)}a)\circ \Ad^{\circ}_{\huaB(b)\huaB(a)}\Ad^{\circ}_{\overline{\huaB(b)}}\bar{b} = \huaB(b)\circ a \circ\huaB(b)\\
&=& b \circ \huaB(b) \circ a \circ \overline{\huaB(b)} \circ \huaB(b) \circ \huaB(a) \circ
\overline{\huaB(b)} \circ \bar{b}  \circ \huaB(b)  \circ \overline{\huaB(a)} \circ \overline{\huaB(b)}\\
&=&  \huaB(b) \circ  a \circ  \overline{\huaB(b)},
\end{eqnarray*}
where in the last equation we used the fact that $a \circ \huaB(a) = e$. Thus,
\begin{eqnarray*}
\huaB (b) \circ a \circ \overline{\huaB(b)} \circ \huaB( \huaB (b) \circ a \circ \overline{\huaB(b)} )&=& \huaB (b) \circ a \circ \overline{\huaB(b)} \circ  \huaB(b \ast a \ast a^{\dagger})\\
&=& \huaB (b) \circ a \circ \overline{\huaB(b)} \circ  \huaB(b) \circ  \huaB(a) \circ \overline{\huaB(b)}= e,
\end{eqnarray*}
which implies that $\huaB(b) \circ a \circ \overline{\huaB(b)}$ is in $K_{-}$. Since $(G, \cdot)$ is an abelian group, we also have
\begin{eqnarray*}
(\huaB (b) \cdot a \cdot \huaB(b)^{-1}) \circ \huaB(\huaB (b) \cdot a \cdot \huaB(b)^{-1} )
= a \circ \huaB(a)=e,
\end{eqnarray*}
which implies that $\huaB (b) \cdot a \cdot \huaB(b)^{-1}$ is in $K_{-}$. Since $K_{-}$ is an ideal of $G_{\huaB}$,  $K_{-}$ is an ideal of $G_{-}$. Likewise we prove that $K_{+}$ is an ideal of $G_{+}$.
\end{proof}

Based on Lemma \ref{lem-fac}, we define a map
$$ \Theta : G_{-}/K_{-} \to G_{+}/K_{+}, \quad \Theta([\huaB(a)]) = [\huaB_{+}(a)], \quad \forall a \in G,$$
where $[\cdot]$ denotes the equivalence class in the two quotients. In order to show that $\Theta$ is well-defined, consider an arbitrary element $a \in K_{-}$, that is, $a = \overline{\huaB(a)} = \huaB(a^{\dagger})$. Note that $a^{\dagger} = \Ad^{\circ}_{\overline{\huaB(a)}}\bar{a} = \bar{a}$. Then we have
\begin{eqnarray*}
\Theta ([ \huaB(b) \circ a])&=& \Theta ([ \huaB(b \ast \bar{a})])\\
&=& [\huaB_+(b \ast \bar{a})]\\
&=& [(b \ast \bar{a}) \circ \huaB(b \ast \bar{a})]\\
&=& [ b \ast \huaB(b) \circ a \circ \overline{\huaB(b)} \circ \huaB(b)\circ\huaB(a^{\dagger}) ]\\
&=& \Theta ([ \huaB(b)]),
\end{eqnarray*}
and
\begin{eqnarray*}
\Theta ([ \huaB(b) \cdot a])&=& \Theta ([ \huaB(b \cdot a^{\dagger})])\\
&=& [\huaB_+(b \cdot a^{\dagger})]\\
&=& [(b \cdot a^{\dagger}) \circ \huaB(b \cdot a^{\dagger})]\\
&\overset{\eqref{enhance-RBO}}{=}& \huaB(b \cdot a^{\dagger}) \cdot (b \cdot a^{\dagger})\\
&=& \huaB(b) \cdot \huaB(a^{\dagger}) \cdot b \cdot a^{\dagger}\\
&=& \Theta ([ \huaB(b)]),
\end{eqnarray*}
which implies that $\Theta$ is well-defined.

\begin{pro}\label{Cayley-trans}
The map $\Theta : G_{-}/K_{-} \to G_{+}/K_{+}$ is an isomorphism of braces, called the {\bf Cayley transform} of the enhanced Rota-Baxter operator $\huaB$.
\end{pro}

\begin{proof}
It is obvious that $\Theta$ is surjective. To see that it is also injective, take $\huaB_{+}(a) = a \circ \huaB(a) \in K_{+}$, that is, $\huaB(a\circ \huaB(a)) = e$. Then we have
\begin{eqnarray*}
\huaB_{+}(\huaB(a)) &=& \huaB(a)\circ \huaB(\huaB(a)) = \huaB(a \ast \huaB(a))\\
  &=& \huaB(a \circ \Ad^{\circ}_{\huaB(a)}\huaB(a)) = \huaB(a \circ \huaB(a)) = e,
\end{eqnarray*}
which implies that $\huaB(a) \in K_{-}$. This proves that $\Theta$ is injective.

We next show that $\Theta$ is a homomorphism of braces, which follows from
\begin{eqnarray*}
\Theta([\huaB(a)]\cdot [\huaB(b)])&=& \Theta([\huaB(a\cdot b)])=[\huaB_+(a\cdot b)]\\
&=& [\huaB_+(a)\cdot \huaB_+(b)]= \Theta([\huaB(a)])\cdot \Theta([\huaB(b)]),
\end{eqnarray*}
and
\begin{eqnarray*}
\Theta([\huaB(a)]\circ [\huaB(b)])&=& \Theta([\huaB(a\ast b)])=[\huaB_+(a\ast b)]\\
&=& [\huaB_+(a)\circ \huaB_+(b)]= \Theta([\huaB(a)])\circ \Theta([\huaB(b)]),
\end{eqnarray*}
according to Proposition \ref{huaB+}.  Therefore, $\Theta$ is an isomorphism of braces.
\end{proof}

Now we consider the direct product brace $(G_{+} \times G_{-}, \cdot_{D},\circ_{D})$, where the multiplications $\cdot_{D}$ and $\circ_{D}$ are given by
\begin{eqnarray*}
(a_{+}, a_{-}) \cdot_{D} (b_{+}, b_{-}) &:=& (a_{+} \cdot b_{+}, a_{-} \cdot b_{-}),\\
(a_{+}, a_{-}) \circ_{D} (b_{+}, b_{-}) &:=& (a_{+} \circ b_{+}, a_{-} \circ b_{-}),
 \quad \forall a_{+}, b_{+} \in G_{+}, a_{-}, b_{-} \in G_{-}.
\end{eqnarray*}
Let $G_{\Theta} \subset G_{+} \times G_{-}$ denote the subset
\[
G_{\Theta} := \{(a_{+}, a_{-}) \in G_{+} \times G_{-} \mid \Theta([a_{-}]) = [a_{+}]\}.
\]
Define a map $\Phi : G \to G_{\Theta}$ by
\[
\Phi(a) := (\huaB_{+}(a), \huaB(a)).
\]

\begin{lem}
With the above notations, $G_{\Theta}$ is a sub-brace of $(G_{+} \times G_{-}, \cdot_{D},\circ_{D})$. Moreover, the map $\Phi$ is an isomorphism of braces from $(G_{\huaB},\cdot, \ast)$ to $(G_{\Theta},\cdot_{D},\circ_{D})$.
\end{lem}

\begin{proof}
By Proposition \ref{Cayley-trans}, for any $(a_{+}, a_{-}), (b_{+}, b_{-}) \in G_{\Theta}$, we have
\begin{eqnarray*}
\Theta([a_{-} \cdot b_{-}]) &=& \Theta([a_{-}] \cdot [b_{-}]) = \Theta([a_{-}]) \cdot \Theta([b_{-}]) = [a_{+}] \cdot [b_{+}] = [a_{+}\cdot b_{+}];\\
\Theta([a_{-} \circ b_{-}]) &=& \Theta([a_{-}] \circ [b_{-}]) = \Theta([a_{-}]) \circ \Theta([b_{-}]) = [a_{+}] \circ [b_{+}] = [a_{+}\circ b_{+}],
\end{eqnarray*}
which implies that $(a_{+} \cdot b_{+}, a_{-}\cdot b_{-}) \in G_{\Theta}$ and $(a_{+} \circ b_{+}, a_{-}\circ b_{-}) \in G_{\Theta}$.  Then $G_{\Theta}$ is a sub-brace of $(G_{+} \times G_{-}, \cdot_{D},\circ_{D})$.

We next check that $\Phi$ is a bijection. Let $a \in G$ such that $\Phi(a) = (e, e)$. Then we have $\huaB(a) = a \circ \huaB(a) = e$. Thus $a = e$, which implies that $\Phi$ is injective. For any $(a_{+}, a_{-}) \in G_{\Theta}$, we have $\Theta([a_{-}]) = [a_{+}]$. Since $a_{-} \in G_{-}$, there exists $a \in G$ such that $\huaB(a) = a_{-}$. Hence we get
\begin{eqnarray*}
\Theta([a_{-}]) = \Theta([\huaB(a)]) = [a \circ \huaB(a)].
\end{eqnarray*}
Therefore, $[a_{+}] = [a \circ \huaB(a)]$, which means that there exists $b \in K_{+}$ such that
$$ a_{+} = a \circ \huaB(a) \circ b. $$
Let $a^{\prime} = a \ast b$. Then we have
\begin{eqnarray*}
\Phi(a^{\prime}) &=& ((a \ast b) \circ \huaB(a \ast b), \huaB(a \ast b))\\
 &=& (a \circ \Ad^{\circ}_{\huaB (a)}(b) \circ \huaB(a) \circ \huaB(b),\huaB(a) \circ \huaB(b))\\
 &=& (a \circ \huaB(a) \circ b, \huaB(a))= (a_+,a_-).
\end{eqnarray*}
Therefore, $\Phi$ is surjective.

Finally, for any $a,b \in G$, by Proposition \ref{huaB+}, we have
\begin{eqnarray*}
\Phi(a \ast b) &=& (\huaB_{+}(a \ast b), \huaB(a \ast b))\\
 &=& (\huaB_{+}(a)\circ \huaB_{+}(b), \huaB(a)\circ \huaB(b))\\
 &=& (\huaB_{+}(a), \huaB(a)) \circ_{D} (\huaB_{+}(b), \huaB(b))\\
 &=& \Phi(a) \circ_{D} \Phi(b).
\end{eqnarray*}
Therefore, $\Phi$ is an isomorphism of braces from $(G_{\huaB},\cdot, \ast)$ to $(G_{\Theta},\cdot_{D},\circ_{D})$.
\end{proof}

\begin{thm}\label{fac-thm-eRB-brace}
{\rm({\bf Factorization Theorem} of enhanced Rota-Baxter two-sided braces)}
Let $(G, \huaB)$ be an enhanced Rota-Baxter two-sided brace. Then every element $a \in G$ can be uniquely expressed as $a = a_{+}\circ \overline{a_{-}}=a_{-}^{-1} \cdot a_{+}$ for $(a_{+}, a_{-}) \in G_{\Theta}$.
\end{thm}

\begin{proof}
For any $a \in G$, we have $a = \huaB_{+}(a) \circ \overline{\huaB(a)}
$ which also equal to $\huaB(a)^{-1} \cdot \huaB_{+}(a)$ by \eqref{enhance-RBO}. To see the uniqueness, if $a = a_{+}\circ \overline{a_{-}} = b_{+}\circ \overline{b_{-}}$, then $$\overline{b_{+}}\circ a_{+} = \overline{b_{-}}\circ a_{-} \in G_{+} \cap G_{-},~\mbox{ and~}~ \Theta(\overline{b_{-}}\circ a_{-}) = \overline{b_{+}}\circ a_{+}.$$ Suppose $\overline{b_{+}}\circ a_{+} = \overline{b_{-}}\circ a_{-} = \huaB(s) \in G_{+} \cap G_{-}$ for some $s \in G$. Then
$$ \Theta([\huaB(s)]) = [\huaB_{+}(s)] = [s \circ \huaB(s)] = [\huaB(s)],$$
which implies that $s\circ \huaB(s) = \huaB(s) \circ c$ for some $c \in K_{+}$. Note that $K_{+} \subset G_{+}$ as an ideal. Then we have $s =\huaB(s) \circ c \circ \overline{\huaB(s)} \in K_{+}$. Therefore, $\huaB(s)= e$. Hence we get $b_{+} = a_{+}$ and $b_{-} = a_{-}$.
\end{proof}

\section{Examples}\label{sec:enumeration}

In this section, we determine Rota-Baxter operators and enhanced Rota-Baxter operators on the two-sided brace induced by the $3$-dimensional Heisenberg Lie algebra $\g$. Moreover, we give the corresponding post-braces and descendent braces. Finally, we  present two Drinfel'd isomorphic solutions $R_1$ and $R_2$ of the Yang-Baxter equation on $\g$ by applying Theorem \ref{post-brace-YBE}.

\begin{ex}{\rm(\cite{GKST})}\label{ex:brace-Lie-alg}
  Let $\g$ be the $3$-dimensional Heisenberg Lie algebra over the field of real numbers $\mathbb{R}$ with the basis $\{e_1,e_2,e_3\}$ satisfying
$$ [e_1,e_2]=e_3,\quad [e_1,e_3]=[e_2,e_3]=0.  $$
Then $(\g,+,\circ)$ is a two-sided brace, where $+$ is the natural abelian group structure on the vector space $\g$ and the group structure $\circ$ is given by
$$ x \circ y:=x+y+\frac{1}{2}[x,y],\quad \forall x,y\in \g.$$
\end{ex}

Since a Rota-Baxter operator  $\huaB:\g \to \g$ on the brace $(\g,+,\circ)$  satisfies $\huaB(x+y)=\huaB(x)+\huaB(y)$, there holds $\huaB(r x)=r\huaB(x)$ for $r\in \mathbb{Q}$. We further assume that $\huaB$ is continuous. Then we have  $\huaB(r x)=r\huaB(x)$ for $r\in \mathbb{R}$, i.e. $\huaB$ is a linear map. 
 With respect to the basis $\{e_1,e_2,e_3\}$ given above, a linear map $\huaB$ is determined by its matrix which we denote by the same symbol, via
$$
\huaB(x)=\begin{pmatrix}
  B_{11} & B_{12} & B_{13} \\  
  B_{21} & B_{22} & B_{23} \\  
  B_{31} & B_{32} & B_{33} 
  \end{pmatrix}\begin{pmatrix}
 x_1  \\  
  x_2 \\  
  x_3 
  \end{pmatrix},\quad \forall x=x_1 e_1+x_2 e_2+x_3 e_3.
$$

\begin{ex}\label{ex:3ex}
There are three classes of Rota-Baxter operators on the two-sided brace $(\g,+,\circ)$ given in Example \ref{ex:brace-Lie-alg} as follows: 
\begin{itemize}
\item[{\rm(i)}]
{\small \[ \left\{ \left .
  \huaB= \begin{pmatrix}
  B_{11} & B_{12} & 0 \\  
  B_{21} & B_{22} & 0 \\  
  B_{31} & B_{32} & 0 
  \end{pmatrix}\,\right|\,
  \begin{vmatrix}
  B_{11} & B_{12} \\
  B_{21} & B_{22}
  \end{vmatrix}=0,\ B_{31},B_{32} \in \mathbb{R} \right\};
  \]
}
\item[{\rm(ii)}]
{\small
	\[ \left\{ \left .
  \huaB=\begin{pmatrix}
  B_{33}+\sqrt{B_{33}(B_{33}+1)} & 0 & 0 \\  
  0 & B_{33}+\sqrt{B_{33}(B_{33}+1)} & 0 \\  
  B_{31} & B_{32} & B_{33} 
  \end{pmatrix}\,\right |\,
  B_{31},B_{32} \in \mathbb{R},B_{33}\in (-\infty,-1)\cup(0,\infty) \right\};
  \]
}
\item[{\rm(iii)}]
{\small \[ \left\{ \left.
  \huaB= \begin{pmatrix}
  B_{33}-\sqrt{B_{33}(B_{33}+1)} & 0 & 0 \\  
  0 & B_{33}-\sqrt{B_{33}(B_{33}+1)} & 0 \\  
  B_{31} & B_{32} & B_{33} 
  \end{pmatrix}\,\right|\,
   B_{31},B_{32} \in \mathbb{R},B_{33}\in (-\infty,-1)\cup(0,\infty) \right\}.
  \]
}
\end{itemize}
Moreover, there is only one class of enhanced Rota-Baxter operators on $(\g,+,\circ)$ given
as follows:
\[ \left\{\left .
  \huaB=\begin{pmatrix}
  0 & 0 & 0 \\  
  0 & 0 & 0 \\  
  B_{31} & B_{32} & 0 
  \end{pmatrix}\,\right|  B_{31},B_{32} \in \mathbb{R} \right\}.
  \]

In fact, by a direct calculation, for all $x=x_1 e_1+x_2 e_2+x_3 e_3, y=y_1 e_1+y_2 e_2+y_3 e_3,$ we have
$$ [x,y]=(x_1 y_2-x_2 y_1)e_3, \quad \forall x,y\in \g, $$
which implies that
$$ x \circ y=x + y+\half(x_1 y_2-x_2 y_1)e_3.   $$
Then we have
$$ \Ad^{\circ}_x (y)=x \circ y \circ (-x)=y+(x_1 y_2-x_2 y_1)e_3.    $$
We consider the equation $\huaB(x) \circ \huaB(y) = \huaB( x \circ (\Ad^{\circ}_{\huaB (x)}y))$ and
obtain the following results by a direct and complicated calculation:
$ B_{13}=B_{23}=0 $ and
\[ \frac{1}{2}
  \begin{vmatrix}
  B_{11}x_1+B_{12}x_2   & B_{11}y_1+B_{12}y_2 \\  
  B_{21}x_1+B_{22}x_2   & B_{21}y_1+B_{22}y_2
  \end{vmatrix}=
  \frac{1}{2} B_{33}
  \begin{vmatrix}
  x_1   & y_1 \\  
  x_2   & y_2
  \end{vmatrix}+B_{33}
  \begin{vmatrix}
  B_{11}x_1+B_{12}x_2   & y_1\\  
  B_{21}x_1+B_{22}x_2   & y_2
  \end{vmatrix}, \quad \forall x_1,x_2,y_1,y_2\in \mathbb{R}.
\]
If $B_{33}\neq 0$ and $B_{33}\in (-\infty,-1)\cup(0,\infty)$, we have $B_{12}=B_{21}=0$ and $B_{11}=B_{22}=B_{33}\pm\sqrt{B_{33}(B_{33}+1)}$
by a direct calculation, which are exactly the second and third classes of Rota-Baxter operators.
If $B_{33}= 0$, we have
\[  \begin{vmatrix}
  B_{11} & B_{12} \\
  B_{21} & B_{22}
  \end{vmatrix}=0,
\]
which is the first class of Rota-Baxter operators.

Moreover, if we consider  enhanced Rota-Baxter operators on the brace $(\g,+,\circ)$ satisfying \eqref{enhanced-RBO}, that is, the following equation holds:
$$(\huaB(x)+z)\circ \huaB(y) = \huaB( x \circ (\Ad^{\circ}_{(\huaB (x) + z)}y)) + z, \quad \forall x,y,z\in \g,$$
then we have
$$ (\huaB(x)+z)\circ \huaB(y)= \huaB(x) \circ \huaB(y)+z+\frac{1}{2}(z_1 \huaB(y)_2-z_2 \huaB(y)_1 )e_3, $$
and
$$ \huaB( x \circ (\Ad^{\circ}_{(\huaB (x) + z)}y)) + z=
\huaB(x \circ \Ad^{\circ}_{\huaB (x)}y )+
\begin{vmatrix}
  z_1 & y_1 \\
  z_2 & y_2
  \end{vmatrix}\huaB(e_3)+z,
$$
which implies that
$ \begin{vmatrix}
  z_1 & y_1 \\
  z_2 & y_2
\end{vmatrix}\huaB(e_3)=\frac{1}{2}
\begin{vmatrix}
  z_1 & \huaB(y)_1 \\
  z_2 & \huaB(y)_2
\end{vmatrix}e_3.
$ According to the above notation, it is equivalent to the following equation:
$$ \begin{vmatrix}
  z_1 & B_{33} y_1-\frac{1}{2}(B_{11} y_1+B_{12} y_2) \\
  z_2 & B_{33} y_2-\frac{1}{2}(B_{21} y_1+B_{22} y_2)
\end{vmatrix}=0, \quad \forall y_1,y_2,z_1,z_2\in \mathbb{R},$$
which implies that $B_{21}=B_{12}=0$ and $B_{33}=\frac{1}{2} B_{22}=\frac{1}{2} B_{11}=0$. Therefore, all enhanced Rota-Baxter operators on the brace $(\g,+,\circ)$ are of the following form
\[ \left\{\left .
  \huaB=\begin{pmatrix}
  0 & 0 & 0 \\  
  0 & 0 & 0 \\  
  B_{31} & B_{32} & 0 
  \end{pmatrix}\,\right|  B_{31},B_{32} \in \mathbb{R} \right\}.
  \]
\end{ex}

\begin{ex}
 Let $\huaB$ be a Rota-Baxter operator on the two-sided brace $(\g,+,\circ)$ given in Example \ref{ex:brace-Lie-alg}. By Proposition \ref{pro:rRBO-post-brace}, we obtain the post-brace $(\g,+,\circ,\rhd_{\huaB})$, where the multiplication $\rhd_{\huaB}$ is given by
$$ x \rhd_{\huaB} y=\Ad^{\circ}_{\huaB(x)} y=y+[\huaB(x),y],\quad \forall x,y\in \g. $$
For the first class of Rota-Baxter operators given in Example \ref{ex:3ex},  the multiplication table for $\rhd_{\huaB}$ is given by

\begin{center}
\centering
\small
\(
\begin{array}{c|ccc}
    \rhd_{\huaB} & e_1 & e_2 & e_3\\ \hline
    e_1 & e_1-B_{21}e_3 & e_2+B_{11}e_3 & e_3 \\
    e_2 & e_1-B_{22}e_3 & e_2+B_{12}e_3 & e_3 \\
    e_3 & e_1 & e_2 & e_3 \\
\end{array}
\)
\end{center}
Then by Corollary \ref{two-descendent-brace}, the corresponding descendent brace $(\g,+,\ast)$ is given by
$$ x \ast y=x \circ (x \rhd_{\huaB} y)=x+y+[\huaB(x),y]+\frac{1}{2}[x,y+[\huaB(x),y]]=x+y+[\huaB(x),y]+\frac{1}{2}[x,y], $$
which is actually a two-sided brace by a direct examination. In the case of the first class of Rota-Baxter operators given in Example \ref{ex:3ex}, the multiplication table for $\ast$ is given by

\begin{center}
\centering
\small
\(
\begin{array}{c|ccc}
    \ast & e_1 & e_2 & e_3\\ \hline
    e_1 & 2e_1-B_{21}e_3 & e_1+e_2+(\frac{1}{2}+B_{11})e_3 & e_1+e_3 \\
    e_2 & e_1+e_2-(\frac{1}{2}+B_{22})e_3 & 2e_2+B_{12}e_3 & e_2+e_3 \\
    e_3 & e_1+e_3 & e_2+e_3 & 2e_3 \\
\end{array}
\)
\end{center}
\end{ex}

\begin{ex}
We consider two non-degenerate involutive solutions $R_1$ and $R_2$ of the Yang-Baxter equation given in Theorem \ref{post-brace-YBE} on the two-sided brace given in Example \ref{ex:brace-Lie-alg}. By a direct calculation, for all $x,y\in \g$, we have
\begin{eqnarray*}
-x+(x \circ y)=-x+x+y+\frac{1}{2}[x,y]=y+\frac{1}{2}[x,y],
\end{eqnarray*}
and
\begin{eqnarray*}
\overline{-x+(x \circ y)}\circ x \circ y=(-y-\frac{1}{2}[x,y])\circ x \circ y
= (x-y)\circ y = x+\frac{1}{2}[x,y],
\end{eqnarray*}
which implies that $$R_1(x,y)=(y+\frac{1}{2}[x,y],x+\frac{1}{2}[x,y]).$$
By a direct calculation, we have
$$x^\dagger=-x+[\huaB(x),x],$$
as the inverse of $x$ in the group $(\g,\ast)$. Then by the fact that $\huaB([x,y])=0$ for all $x,y\in \g$, we have
\begin{eqnarray*}
-x+x \ast y=-x+x+y+ \frac{1}{2}[x,y]+[\huaB(x),y]=y+ \frac{1}{2}[x,y]+[\huaB(x),y],
\end{eqnarray*}
and
\begin{eqnarray*}
(-x+x \ast y)^\dagger \ast x \ast y
&=& (-y-\frac{1}{2}[x,y]-[\huaB(x),y]+[\huaB(y),y])\ast x \ast y \\
&=& (x-y-[\huaB(x),y]+[\huaB(y),y]-[\huaB(y),x])\ast y \\
&=& x+\frac{1}{2}[x,y]+[x,\huaB(y)],
\end{eqnarray*}
which implies that $$R_2(x,y)=(y+ \frac{1}{2}[x,y]+[\huaB(x),y],x+\frac{1}{2}[x,y]+[x,\huaB(y)]).$$
In this case, $R_2$ can be expressed in terms of $\ast$ as follows:
$$ R_2(x,y)=(-x+x \ast y,y+(-y)\ast x), \quad \forall x,y\in \g. $$
Then by Theorem \ref{post-brace-YBE}, for all $x,y\in \g$, we have
\begin{eqnarray*}
  (L^\blacktriangleright_x)^{-1}(-x+ (x\circ y))
&=& x^\dagger \ast (x \circ y)\\
&=& (-x+[\huaB(x),x]) \ast (x+y+\frac{1}{2}[x,y])\\
&=& -x+[\huaB(x),x]+x+y+ \frac{1}{2}[x,y]-\frac{1}{2}[x,y]-[\huaB(x),x]-[\huaB(x),y]\\
&=& y-[\huaB(x),y],
\end{eqnarray*}
which implies that $R_1$ is Drinfel'd isomorphic to $R_2$ via the map $\overline{\omega}:\g \times \g \to \g \times \g$ given by
$$ \overline{\omega}(x,y)=(x,y-[\huaB(x),y]), \quad \forall x,y\in \g.$$
\end{ex}

\noindent
{\bf Acknowledgements.} This research is supported by NSFC (12471060, W2412041, 123B2017).

\smallskip
\noindent
{\bf Declaration of interests. } The authors have no conflicts of interest to disclose.

\smallskip
\noindent
{\bf Data availability. } Data sharing is not applicable as no data were created or analyzed.

\end{document}